\documentclass[10pt,a4paper]{article}

\usepackage[utf8]{inputenc}

\usepackage{fullpage}
\usepackage{amsmath}
\usepackage{amsfonts}
\usepackage{amssymb}
\usepackage{amsthm}
\usepackage{times}
\usepackage{dsfont}
\usepackage{bbm}
\usepackage{color}
\usepackage{authblk}
\usepackage{paralist}
\usepackage{algorithm}
\usepackage{algpseudocode}
\usepackage{quantum}
\usepackage[final]{graphicx}
\usepackage[colorlinks = True, linkcolor=blue, citecolor=blue, pdfpagemode=UseNone, pdfstartview=FitH, hypertexnames=false]{hyperref}
\usepackage[capitalise]{cleveref}
\usepackage{todonotes}

\newcommand{\affilcr}{\protect\\}

\algblockx[forall]{Forall}{EndForall}%
  [1]{\textbf{for all} #1 \textbf{do}}{\textbf{end}}
\algblockx[loop]{Loop}{EndLoop}%
  [1]{\textbf{loop} #1}{\textbf{end loop}}

\crefrangelabelformat{equation}{\textup{(#3#1#4)}--\textup{(#5#2#6)}}
\creflabelformat{enumi}{(#2#1#3)}
\Crefname{algorithm}{Algorithm}{Algorithms}
\Crefname{table}{Table}{Tables}
\Crefname{remark}{Remark}{Remark}
\Crefname{theorem-break}{Theorem}{Theorems}
\Crefname{corollary-break}{Corollary}{Corollaries}
\Crefname{definition-break}{Definition}{Definitions}
\crefname{part}{part}{parts}  
\Crefname{part}{Part}{Parts}
\crefname{equations}{Eqs.}{Eqs.}
\Crefname{equations}{Equations}{Equations}
\creflabelformat{equations}{\textup{(#2#1#3)}}
\crefrangelabelformat{equations}{\textup{(#3#1#4)}--\textup{(#5#2#6)}}

\newcommand{\hide}[1]{}

\def \be   {\begin{equation}}
\def \ee   {\end{equation}}
\def \ben  {\begin{enumerate}}
\def \een  {\end{enumerate}}
\def \bi   {\begin{itemize}}
\def \ei   {\end{itemize}}

\newtheorem{theorem}{Theorem}
\newtheorem{lemma}[theorem]{Lemma}

\begin{document}

\title{An Information-Theoretic Proof of the Constructive\\
  Commutative Quantum Lov\'{a}sz Local Lemma}

\author[1]{Martin Schwarz\thanks{m.schwarz@univie.ac.at}}
\affil{Vienna Center for Quantum Science and Technology,
  \affilcr Faculty of Physics, University of Vienna, A-1090 Wien, Austria}
\author[2]{Toby S. Cubitt\thanks{tsc25@cam.ac.uk}}
\affil[2]{DAMTP, University of Cambridge,
  Centre for Mathematical Sciences,\affilcr
  Wilberforce Road, Cambridge CB3 0WA, United Kingdom}
\author[1,3]{Frank Verstraete\thanks{frank.verstraete@univie.ac.at}}
\affil[3]{Faculty of Science, Ghent University, B-9000 Ghent, Belgium}
\date{}

\maketitle

\begin{abstract}
  The Quantum Lovász Local Lemma (QLLL) \cite{AKS12} establishes non-constructively that any quantum system constrained by a local Hamiltonian has a zero-energy ground state, if the local Hamiltonian terms overlap only in a certain restricted way. In this paper, we present an efficient quantum algorithm to prepare this ground state for the special case of commuting projector terms. The related classical problem has been open for more than 34 years. Our algorithm follows the breakthrough ideas of Moser's \cite{moser09} classical algorithm and lifts his information theoretic argument to the quantum setting.
  A similar result has been independently published by Arad and Sattath \cite{AS13} recently.
\end{abstract}

\section{Introduction}\label{intro}
In 1973 László Lovász proved a remarkable probabilistic lemma nowadays known as \emph{Lovász Local Lemma (LLL)} \cite{EL75, spencer1977} . Informally, it says that whenever events in a set of probability events are only locally dependent (i.e. each event depends on at most a constant number of other events), then with positive probability \emph{none} of them occurs. This probability might be extremely small, nevertheless the lemma shows that such an event \emph{exists}. Lovász and Erd\H{o}s applied this lemma with great success to prove the existence of various rare combinatorial objects, an approach which came to be known as the \emph{probabilistic method} \cite{AS00}. Their method has one drawback: even though the LLL shows the existence of certain objects, it doesn't provide any clue of how to \emph{construct} such objects efficiently -- the lemma is non-constructive. Things started to change when in 1991 Beck was the first to give an efficient algorithm to construct such objects, but only under assumptions stronger than the LLL \cite{beck91}. After a sequence of improvements on Beck's work, Moser's breakthrough in 2009 finally gave us a constructive and efficient proof of the LLL under the same assumptions as the original one \cite{moser09}. First, he proved a widely used variant called the \emph{symmetric LLL}, and then jointly with Tardos gave a fully general constructive and efficient proof of the LLL \cite{MT10}. The symmetric LLL considers the special case, where the probabilities of all of the dependent events are bounded by the same constant, and can be stated as follows:
\begin{lemma}[Symmetric Lovász Local Lemma]
Let $A_1,A_2,...,A_m$ be a set of events such that each event occurs with probability at most $p$.
If each event is independent of all others except for at most $d-1$ of them, and
\[
epd \leq 1,
\]
then
\[
\Pr\left[\;\overline{A_1} \cap \overline{A_2} \cap ... \cap \overline{A_m}\;\right] > 0.
\]
\end{lemma}
The symmetric LLL is often used in the context of constraint satisfaction problems (CSPs) to prove
the existence of an object specified by a list of local constraints.
In this case one considers, say, $n$ bit strings $X$ chosen uniformly at random. The events are given by the local constraint functions $A_i=f_i(X)$, where each function $f_i$ is $k$-local in the sense that it depends only on $k$ of the $n$ bits; the event occurs if the constraint is satisfied. If these $A_i$ meet the  constraints of the symmetric LLL, the LLL implies that an $x$ satisfying all the constraints \emph{exists}, and Moser's algorithm can be used to construct such an $x$ efficiently. In this way the LLL also implies that the set of $k$-SAT instances, where each variable occurs in at most $d < 2^k/ek$ clauses, is \emph{always} satisfiable. (Without this restriction on variable occurrence, deciding satisfiability is of course the archetypical \textsf{NP}-complete problem.)

During the STOC 2009 presentation of his result, Moser presented a beautiful information-theoretic argument, valid under very slightly stronger conditions, which underlies the more complicated but tight result in \cite{moser12}. It is this argument that the present paper generalizes to the quantum setting.

The non-constructive proof of the LLL has recently been generalized to the quantum case by Ambainis, Kempe, and Sattath \cite{AKS12}. In this setting, events are replaced by orthogonal projectors of rank $1$ (or rank $r$ in general) onto $k$-local subsystems, and the authors achieve a non-constructive proof of a Quantum Lovász Local Lemma (QLLL) with exactly the same constants as in the classical version.

\begin{lemma}[Symmetric Quantum Lovász Local Lemma \cite{AKS12}] \label{lem:sQLLL}
Let $\{\Pi_1,...,\Pi_n\}$ be a set of $k$-local projectors of rank at most $r$. If every
qubit appears in at most $d < 2^k/(e\cdot r \cdot k)$ projectors, then the instance is satisfiable.
\end{lemma}

In this paper, we generalize Moser's algorithm to the quantum setting in the special case of commuting projectors, yielding an efficient proof of \cref{lem:sQLLL} for this case. While all of our projectors are diagonal in a common basis, the basis vectors will in general be highly entangled quantum states. The (classical) constructive LLL does not immediately apply in the diagonal basis. Indeed, the preparation of such highly entangled ground states is far from trivial and subject to active research in the field of quantum Hamiltonian complexity theory \cite{osborne12,schuch11,AE11,STVPGC13}.

Furthermore, we improve upon Moser's argument and make it tight up to the assumptions of the non-constructive symmetric (Q)LLL. Of course, this also implies a tight algorithmic result for the classical special case. Our argument relies on a simple universal method to compress a binary classical bit sequence, which yields the tight result. In the process of generalizing the result to the quantum setting, we explicitly bound the run-time and error probabilities using (a tight special case of) the \emph{strong converse of the typical subspace theorem} \cite{winter99} as an indispensible ingredient, which is a fundamental result of quantum information theory.

More precisely, we prove the following efficient symmetric Quantum Lovász Local Lemma for commuting projectors with the same parameters as the original LLL and QLLL. Our proof is a quantum information-theoretic argument, but by restricting to classical constraints our argument immediately specializes to a tight classical information-theoretic proof.

\begin{theorem}[Efficient symmetric commutative QLLL] \label{thm:escQLLL}
  Let $\Pi_1,\Pi_2,...,\Pi_m$ be a set of commuting $k$-local projectors of rank at most $r$ acting on a system of $n$ qubits. If each projector intersects with at most $d-1$ of the others, where $d \leq \frac{2^k}{r e}$, then for any $\varepsilon > 0$ there exists a quantum algorithm with run-time $O\!\left(m+\log(\frac{1}{\varepsilon})\right)$ that returns a quantum state $\sigma$ with probability $1-\varepsilon$, such that $\sigma$ has energy zero, i.e. $\forall i, 1\leq i \leq m: tr(\Pi_i \sigma)=0$.
\end{theorem}
It might be interesting to note that for non-commuting projectors our proof still implies that the algorithm terminates within the same run-time bound, but the argument about the energy of the state returned (\cref{lem:recursion}) is no longer applicable. \Cref{lem:recursionsimple} (see also \cref{lem:recursion}) is the crucial and only place in the proof where commutativity of the projectors is used.

In \cref{sec:algorithm}, we fix the notation and review Moser's classical algorithm. In \cref{sec:quantum_algorithm} we present the key ideas of our quantum generalization, and give a simple quantum information-theoretic analysis in \cref{sec:analysis} which leads to the main result. (A manifestly unitary variant of the recursive algorithm, complete with technical details, is given in \cref{sec:details}.) We conclude in \cref{sec:conclusion}.

\section{The Algorithm} \label{sec:algorithm}
In this section we describe our quantum version of Moser's algorithm. Before we do so, we quickly review Moser's classical original algorithm. We will start by setting up some notation, where we try to keep the notational differences between the quantum and classical case at a minimum.

The input to the classical (quantum) algorithm consists of a $k$-(Q)SAT instance. Each $k$-(Q)SAT instance is defined on $n$ (qu)bits and consists of $m$ clauses (projectors of rank at most $r$) $\{\Pi_i\}_{1\leq i \leq m}$. Each clause (projector) is $k$-local, i.e. it acts non-trivially only on a subset of $k$ (qu)bits and as the identity on the $n-k$ remaining qubits. Given an instance $\{\Pi_i\}$, the \emph{exclusive neighborhood function} $\Gamma(\Pi_i)$ returns an ordered tuple of projectors sharing at least one qubit with $\Pi_i$. Furthermore we define the inclusive neighborhood function $\Gamma^+(\Pi_i) = \Gamma(\Pi_i) \cup \Pi_i$. The $j^{\text{th}}$ neighbor of $\Pi_i$ is then defined as $\Gamma^+(\Pi_i)_j$. To simplify the notation, we sometimes write $\Gamma^+(i,j)$ instead of $\Gamma^+(\Pi_i)_j$. In the special case of a $k$-QSAT instance where all $\{\Pi_i\}$ are diagonal in the standard basis, it reduces to a classical $k$-SAT instance and projectors reduce to clauses. All logarithms in this paper use base $2$.

\begin{algorithm}[!hbtp]
  \caption{Classical and quantum information-theoretic LLL solver}\label{alg:moser}
  \begin{algorithmic}[1]
    \Procedure{\textnormal{solve\_lll}}{$\Pi_1,\Pi_2,\dots,\Pi_m$}
      \State $W \gets n$ uniformly random bits \Comment{initial state}
      \State $R \gets kN$ uniformly random bits \Comment{source of randomness}
      \State $t \gets 0$, $L \gets 0...0$ \Comment{book keeping registers}
      \For{$i \gets 1$ to $m$}
        \State fix$(\Pi_i)$ \label{alg:fix_all}
      \EndFor
      \State return (SUCCESS, W)
    \EndProcedure
    \Procedure{\textnormal{fix}}{$\Pi_i$}
      \State measure $\Pi_i$ on $W$  \label{ln:moser:measure}
      \State append the binary result to the execution log, $L$
      \If{$\Pi_i$ was violated}
        \State swap subsystem of $\Pi_i$ with block $t$ in $R$
		\State apply $U_i$ to rotate the state of the swapped subsystem in $R$
        \State $t \gets t+1$
        \Forall{$\Pi_j \in \Gamma^+(\Pi_i)$}
          \State fix$(\Pi_j)$
        \EndForall
      \EndIf
    \EndProcedure
  \end{algorithmic}
\end{algorithm}

\subsection{Moser's classical algorithm} \label{sec:classical_algorithm}
We will now quickly review Moser's classical algorithm to set the scene for our quantum generalization. In \cref{alg:moser} we assume a classical $k$-SAT instance as input. The algorithm operates on a register of $n$ bits sampled from a uniformly random source. The main procedure \emph{solve\_lll()} iterates over the clauses $\Pi_1,\Pi_2,\dots,\Pi_m$ and calls subroutine \emph{fix($\Pi_i)$} on each. Procedure \emph{fix($\Pi_i$)} checks if $\Pi_i$ is satisfied, records the outcome to a logging register $L$ (``the log'') and returns if it is. Otherwise \emph{fix()} resamples the bits of the unsatisfied clause from the uniformly random source and recurses on all neighbors in $\Gamma^+(\Pi_i)$ in turn.
Throughout the paper \emph{fixing a clause} or \emph{fixing a projector} will mean entering such a recursion. Whenever we observe a clause not to be satisfied, we say the measurement of the clause has \emph{failed} (or \emph{succeeded} otherwise.) In the quantum case, whenever a projective measurement $\{\Pi_i, (\id-\Pi_i)\}$ has outcome $\Pi_i$ we say it has failed (or \emph{succeeded} if the outcome is $(\id - \Pi_i)$.)

Moser's key insight was to understand \cref{alg:moser} as a compression algorithm, that draws entropy from a uniformly random source and compresses it into the log register $L$. He shows that the random initial state of $n$ bits and all entropy drawn from the source during execution of the algorithm can be losslessly compressed into the log and the output state. By showing that each failed measurement yields a tighter bound on the entropy of the system, he argues that the algorithm must terminate with high probability after $O(m)$ measurements, as otherwise the entropy of the system was compressed below the entropy drawn from the source.
Furthermore, each time $\emph{fix()}$ returns, one more projector is satisfied. Thus, once the algorithm terminates, all projectors are satisfied and the output state must therefore have energy zero.

In Moser's algorithm the log is introduced merely as a bookkeeping device to facilitate the correctness proof of the algorithm. It is not necessary to produce the log in ``real world'' implementations; the log is merely a proof device to allow one to argue about the entropy of the system by constructing a reversible compression scheme. Since a quantum algorithm in the standard quantum circuit model is unitary, thus in particular reversible, and the concept of reversible lossless compression is central to Moser's proof, this proof approach is a natural fit, and an ideal starting point to develop an efficient quantum algorithm for the QLLL based on a quantum information-theoretic argument. Once unitarity is \emph{required}, the log is no longer an optional, fictitious device. Instead, it becomes a natural and \emph{necessary} by-product of any unitary (or even reversible) implementation.

\subsection{The quantum algorithm} \label{sec:quantum_algorithm}
Although we have to modify the analysis somewhat, our quantum algorithm is just a coherent version of the original classical algorithm of \cref{alg:moser}. In this section, we show how a beautifully simple quantum information-theoretic analysis of this coherent algorithm gives the desired result. A fully detailed version of the proof based on a manifestly unitary version of \cref{alg:moser} (i.e.\ \cref{alg:quantummoser}) is given in \cref{sec:details}.

Unsurprisingly, the quantum algorithm operates on four registers: an $n$-qubit work register $W$, an $T$-qubit log register $L$ consisting of qubits labeled $j_1,...,j_T$, a $kN$-qubit randomness register $R$, and a $\log N$-qubit register $t$ counting the number of failed measurements.\footnote{The algorithm will also have to store some additional data for classical book-keeping, which however we neglect here as it isn't important in the analysis. Full details are given in \cref{sec:details}.} Henceforth, $j_l$ will denote the $l^{\text{th}}$ qubit of $L$, and $R_t$ will denote the $t^{\text{th}}$ \emph{block} of $k$ qubits in $R$. We will use $W_i$ to denote the $k$ qubits in $W$ on which the $i^\text{th}$ projector acts non-trivially. We use $\Pi_i$ to denote both the projector on $W_i$, and the projector $\Pi_i\ox\id$ extended to the whole of $W$; when not indicated explicitly, it will be clear from context which we mean. We initialise the quantum registers to the state
\begin{equation}\label{eq:initialstate}
\ket{\psi_0^{x,y}} = \ket[W]{x} \ket[R]{y} \ket[L]{0_1,...,0_T} \ket[t]{0}
\end{equation}
where $x,y$ are uniformly random bit strings of sizes $n$ and $kN$, respectively.
\Cref{alg:moser} proceeds by \emph{coherently measuring} projectors on the work register and appending the measurement outcomes to the log register. More precisely, a ``coherent measurement'' of $\Pi_i$ is the following unitary operation between the work register and the next unused qubit in the log register.\footnote{The algorithm necessarily keeps track of the index of the next unused log register qubit, as part of the classical bookkeeping implicit in \cref{alg:moser}.}:
\begin{equation}
  C_i = \Pi_{W_i} \ox X_{j_l} + (\id - \Pi_i)_{W_i} \ox \id_{j_l}.
\end{equation}
If $l-1$ measurements have been performed so far, the next coherent measurement writes its outcome to the $l^{\text{th}}$ qubit in the log register $L$.

As is well known~\cite{NC00,wilde13}, when applied to an arbitrary state of the work register $W$ and a $\ket[j_l]{0}$ in the log register $j_l$, the unitary $C_i$ prepares a coherent superposition of the two measurement outcomes in the log register $j_l$, entangled with the corresponding post-measurement state in the work register. The square-amplitudes of the two components are the probabilities of the corresponding measurement outcomes.

If a projector $\Pi_i$ is violated (outcome ``1''), we know that the state of the subsystem $W_i$ is contained in the subspace $\Pi_i$. In this case, we proceed by taking the next $k$ qubits from the randomness register, and swapping them with the $k$ work-qubits we've just measured. The state of the measured qubits must be in the $r$-dimensional subspace projected onto by $\Pi_i$. We can therefore apply a unitary $U_i$ to rotate the measured qubits (which are now in the randomness register) into a fixed $r$-dimensional subspace which is independent of the particular $\Pi_i$ measured. We identify this subspace with the rank-$r$ projector $P_r=diag(1,...,1,0,...,0)$. The unitary $U_i$ can be computed classically for each $i$ by diagonalizing $\Pi_i$, i.e.\ $U_i \Pi_i U_i^\dag = P_i \leq P_r$ with equality if $rk(P_i)=rk(P_r)=r$. Let us denote this sequence of unitary swap-and-rotate operations as $R_i$. Note that the measured, swapped, and rotated $k$ qubits $\ket{\varphi_i}$ have support on subspace $P_r$ only, since
\begin{equation} \label{eq:livesinPr}
  U_i \Pi_i \ket[R]{\varphi_i}
  = U_i \Pi_i U_i^\dag U_i \ket[R]{\varphi_i}
  = P_r U_i \ket[R]{\varphi_i}
\end{equation}
The following partial isometry implements this swap-and-rotate procedure (it can be extended to a unitary in the usual way):
\begin{equation} \label{eq:defRi}
  R_i = \proj{1}^{j_l} \ox (\id^{W_i}\ox U_i^{R_t}\cdot U_{\text{SWAP}}^{W_iR_t})
                 + \proj{0}^{j_l} \ox \id^{W_iR_t}.
\end{equation}
We will always apply $R_i$ immediately after each coherent measurement, so for brevity we refer to the whole isometry $R_iC_i$ as a ``measurement operation''. Whenever we get a violation, we increment the boolean count register $t$.

The recursive algorithm now proceeds analogously to the classical algorithm \cref{alg:moser}. The only differences are that we interpret $\Pi_i$ as commuting projectors (not necessarily diagonal in the computational basis), and that `measure' in \cref{ln:moser:measure} is interpreted as a coherent measurement causing the state (and thus the control flow) to split into a superposition depending on the measurement outcomes.\footnote{For an explicit, manifestly unitary description that includes \emph{all} the classical bookkeeping in the quantum description, see \cref{alg:quantummoser} in \cref{sec:details}.}

Note that any computational basis state describing a sequence of measurement outcomes \emph{uniquely determines} the next measurement to perform; i.e.\ there is a deterministic function $f:\{0,1\}^*\mapsto \{[m],\bot\}$ from finite sequences $j_1,\dots,j_{l-1}$ of previous measurement outcomes to the index $i_l$ of the next measurement (i.e. $i_l=f(j_1,\dots,j_{l-1})$). If there is no further measurement to perform ($\bot$), the measurement sequence terminates (i.e.\ $f(\dots,\bot)=\bot$). It is not difficult to see that this function can be computed efficiently classically. By linearity, we can extend this to a unitary operation on arbitrary superpositions of a specific number of measurement outcomes.

Apart from the measurement operations, the rest of the algorithm involves purely classical processing to determine the next measurement, and thus \emph{is diagonal in the computational basis}. Furthermore, each measurement operation acts on a fresh log qubit. Thus orthogonal states of the log remain orthogonal for the rest of the computation. This allows us to view the execution of the algorithm as a coherent superposition of \emph{histories}, which may be analyzed independently. \Cref{lem:superposition} in \cref{sec:details} makes this precise, and shows that after $T$ coherent measurements, the state (essentially) has the form
\begin{align} \label{eq:superposition3}
\ket{\psi_T^{x,y}} &=\sum_{j_1,\dots,j_T \in \{0,1\}}  P_r^{\ox t_{j_1,\dots,j_T}} \ket[W,R]{\varphi_{j_1,\dots,j_T}}\ket[L]{j_1,\dots,j_T} \ket[t]{t_{j_1,\dots,j_T}}.\quad
\end{align}
Henceforth, we refer to any term in \cref{eq:superposition3} indexed by $j_1,\dots,j_T$ as a \emph{history}. Note the tensor product
structure among the registers in each history.

We let the algorithm run for a total of $T=m+Nd$ measurement steps, for some $N$ chosen in advance. If the recursion in \cref{alg:moser} has reached a maximum of $N$ failed measurements or terminates early, the algorithm (coherently) does nothing for the remaining steps. Finally, after running for this many steps, we measure the log register $L$ in order to collapse the superposition of measurement outcomes to a particular measurement sequence.

\section{Analysis} \label{sec:analysis}
To show that our algorithm efficiently finds a state in the kernel of all $\Pi_i$ with high probability, we need to prove two properties captured in \cref{lem:recursionsimple} and \cref{thm:success}, that together imply the desired result:
\begin{enumerate}[(1).]
  \item If the sequence of measurement outcomes terminates, the corresponding state of the work register is in the kernel of all $\Pi_i$ (\cref{lem:recursionsimple}). \label{it:recursion_success}
  \item The probability that the measurement sequence terminates goes exponentially to~1 for $N > m/(k-\log(der))$ (\cref{thm:success}). \label{it:exp_success_prob}
\end{enumerate}
\begin{lemma}\label{lem:recursionsimple}
  Let $\ket[W]{\varphi_{l}} = \ket[W]{\varphi_{j_1,\dots,j_l}}$ be the state of register $W$ in a history where the algorithm has obtained a failure in the $l^{\text{th}}$ measurement outcome, thereby starting a recursion. Assuming that the recursion eventually terminates, let $\ket{\varphi_m}=\ket{\varphi_{j_1,\dots,j_m}}$ be the state of register $W$ when the algorithm has just returned from that recursion after measurement $m \geq l+k$. Then
  \begin{enumerate}[(i).]
    \item all satisfied projectors $\Pi_i$ stay satisfied, i.e. if $\Pi_i \ket{\varphi_{l}} = 0$, then also $\Pi_i \ket{\varphi_{m}} = 0$. \label{it:other_projectors_satisfied3},
    \item the originally unsatisfied projector $\Pi_l$ is now satisfied, i.e. if $\Pi_{l}\ket{\varphi_{l}}=\ket{\varphi_{l}}$, then $\Pi_{l}\ket{\varphi_{m}}=0$. \label{it:current_projector_satisfied3}
    \end{enumerate}
\end{lemma}
\begin{proof}
  We prove \Cref{lem:recursionsimple} by induction on the recursion level $s$. Let $\Pi_{s}$ be the projector that shall be fixed in the level $s$.

  \textbf{Base case:} Consider the deepest level of recursion, which necessarily exists since, by assumption, the recursion eventually terminates. After the failed $\Pi_{l}$ measurement, the algorithm performs the swap-and-rotate operation followed by measurements of all projectors in $\Gamma^+(\Pi_{l})$ on the state $\Pi_{l} \ket{\varphi_{l}}$. These must succeed, since the algorithm is already at the deepest level of recursion. Thus the algorithm returns yielding the state $\ket{\varphi_m}$. Since $\Pi_{l} \in \Gamma^+(\Pi_{l})$ and all $\Pi_{l}$ commute, \labelcref{it:current_projector_satisfied3} follows. To show \labelcref{it:other_projectors_satisfied3}, note that all previously satisfied $\Pi_i \in \Gamma^+(\Pi_{l})$ clearly stay satisfied, i.e.\ $\forall \Pi_i \in \Gamma^+(\Pi_{l}): \Pi_i \ket{\varphi_m}=0$. For all other $\Pi_i \notin \Gamma^+(\Pi_{l})$, notice that $\Pi_i$ commutes with the swap-and-rotate operation as they act on disjoint subsystems, yielding $\forall \Pi_i \notin \Gamma^+(\Pi_{l}): \Pi_i \ket{\varphi_m}=0$, which proves the base case.

  \textbf{Inductive step:} As induction hypotheses, assume \labelcref{it:other_projectors_satisfied3,it:current_projector_satisfied3} are true for any originally unsatisfied projector $\Pi_{s+1}$ after the algorithm returns from recursion level $s+1$. At level $s$ of the recursion, after a failed measurement $\Pi_{l}$ the algorithm performs the swap-and-rotate operation followed by measurements of all projectors in $\Gamma^+(\Pi_{l})$ on the state $\Pi_{l} \ket{\varphi_{l}}$. For any failed measurement, the algorithm will recurse to level $s+1$ and return with \labelcref{it:other_projectors_satisfied3,it:current_projector_satisfied3} satisfied by the induction hypothesis. Thus, after returning from the recursion, one additional $\Pi_i \in \Gamma^+(\Pi_l)$ is satisfied. For any successful measurement, again one additional $\Pi_i$ is satisfied due to commutativity of the $\Pi_i$. Thus, once the iteration over the neighborhood is complete, the algorithm returns the state $\ket{\varphi_m}$ with all $\Pi_i \in \Gamma^+(\Pi_l)$ satisfied. Since $\Pi_{l} \in \Gamma^+(\Pi_{l})$, \labelcref{it:current_projector_satisfied3} follows. To see that \labelcref{it:other_projectors_satisfied3} also holds, note that all previously satisfied $\Pi_i \in \Gamma^+(\Pi_{l})$ stay satisfied, i.e.\ $\forall \Pi_i \in \Gamma^+(\Pi_{l}): \Pi_i \ket{\varphi_m}=0$. For all other $\Pi_i \notin \Gamma^+(\Pi_{l})$, notice that $\Pi_i$ commutes with the swap-and-rotate operation as they act on disjoint subsystems, yielding $\forall \Pi_i \notin \Gamma^+(\Pi_{l}): \Pi_i \ket{\varphi_m}=0$. This establishes the inductive step, and the lemma follows.
\end{proof}

\noindent Property~\labelcref{it:recursion_success} follows from \cref{lem:recursionsimple} and the fact that the algorithm measures each projector $\Pi_i$ once at the top level of the recursion. Property~\labelcref{it:exp_success_prob} is the content of the following lemma.

\begin{lemma}\label{thm:success}
  If we let the algorithm run for $T=m+Nd$ steps, the probability that the measurement sequence terminated within this number of steps is $\geq 1 - 2^{-N(k-\log der) + m + \log N}$.
\end{lemma}

\begin{proof}
  The proof rests on three simple facts: (i)~The initial state is maximally-mixed on $n + kN$ qubits (tensor a pure state on the rest). (ii)~The algorithm is unitary. (iii)~If a total of $M$ violations occurred, the information stored in the log register $L$ can be compressed to $m + M\log(de)$ qubits.

  Consider a computational basis state $\ket[L]{\sigma}\ket[t]{M}$ of the log and count registers, describing a particular (classical) history $\sigma$ with a total of $M$ violations. Since the count register is incremented each time the algorithm measures a violation, $\sigma$ must contain exactly $M$ $1$s. By encoding $\sigma$ as the \emph{index} $\iota(\sigma;M)$ of $\sigma$ in the lexicographically-ordered set of length-$N$ bit strings that contain exactly $M$ ones, we could losslessly and deterministically compress $\sigma$ to $m + M\log(de)$ bits~\cite[Ch.~13.2]{CT06}. (Note that we do not need to actually perform this compression step as part of the algorithm; it is sufficient that it is possible.) By linearity, we can extend this lossless compression to a unitary operation on the log and count registers $L$ and $t$:
  \begin{equation}
    U_C \ket[L]{\sigma}\ket[t]{M}
    = 
      \underbrace{\ket{\iota(\sigma;M)}\ket{0}^{\ox (T - m - M\log(de))}}_L
      \underbrace{\ket{M}}_t.
  \end{equation}

  Furthermore, since $M$ violations occurred, each of the $M$ subsystems of $k$ qubits in the register $R$ only has support on the $r$-dimensional subspace $P_r$ (in the respective subsystem) by \cref{eq:livesinPr}. Given this, if we apply $U_C$ to the state of the log and count registers $L$ and $t$, the following projector projects onto measurement histories with $M=N$:
  \begin{samepage}
  \begin{gather}
    P_N = U_C^\dag P U_C,\\
    \intertext{where}
    P = \underbrace{ \id^{\otimes n} \vphantom{\ket{0}^{\ox k}} }_W
        \ox \underbrace{ {\phantom{|}P_r}^{\otimes N} }_R
        \ox \underbrace{ \id^{\ox m + N\log de} \ox \proj{0}^{\ox(T-m-N\log de)} }_L
        \ox \underbrace{ \id^{\ox \log N} \vphantom{\ket{0}^{\ox k}} }_t.
  \end{gather}
  \end{samepage}

  \noindent Meanwhile, from \cref{eq:initialstate}, the initial state of the registers is
  \begin{equation}
    \rho_0 = \frac{1}{2^{n+kN}}\sum_{x,y} \ket{\psi_0^{x,y}}\bra{\psi_0^{x,y}} = \frac{\id_W}{2^n} \ox \frac{\id_R}{2^{kN}}
      \ox \proj[L]{0} \ox \proj[l]{0} \proj[t]{0}.
  \end{equation}
  Let $U$ denote the overall unitary describing the algorithm. The probability of measuring $P_N$ on the final state of the algorithm is then\footnote{Note that this inequality is none other than a sharp version of the strong converse of the typical subspace theorem \cite[Lemma I.9]{winter99}, for the simple case of the maximally mixed state.}
  \begin{equation}
    \begin{split}
      \tr[P_N U\rho_0 U^\dag]
        &= 2^{-n-kN} \tr\left[P U_C U(\id_{WR}\ox\proj[L,l,t]{0})U^\dag U_C^\dag\right]\\
        &\leq 2^{-n-kN} \tr P
        = 2^{-N(k-\log der) + m + \log N}.
    \end{split}
  \end{equation}

  Now, any measurement sequence where \emph{less} than $N$ measurements failed must have terminated early, since the total number of measurement steps $T=m+Nd$ is clearly sufficient to return from any recursion with less than $N$ failed measurements (cf.\ \cref{alg:moser}). Thus the projector $\id - P_N$ projects onto histories in which the sequence of measurement outcomes terminated, and the lemma follows.
\end{proof}

Choosing $N=O\left(\frac{m+\log(\frac{1}{\varepsilon})}{k-\log(der)}\right)$ in \cref{thm:success} suffices to produce the desired output state in register~$W$ with success probability $1-\varepsilon$. Together with \cref{lem:recursionsimple}, this proves Properties~\labelcref{it:recursion_success,it:exp_success_prob}, and hence \cref{thm:escQLLL}.

\section{Conclusions} \label{sec:conclusion}
We have presented a quantum generalization of Moser's algorithm and information-theoretic analysis to efficiently construct a zero-energy ground state of certain local Hamiltonians. The existence of such ground states has been established by the non-constructive Quantum Lovász Local Lemma \cite{AKS12}. Our algorithm requires the additional assumption that the Hamiltonian is a sum of commuting projectors. In fact, for this special case, our algorithm is a \emph{constructive proof} of the Quantum Lovász Local Lemma, as our argument does not depend on the non-constructive result of \cite{AKS12}. After completion of this work, we have learned about a similar result of Arad and Sattath \cite{AS13}. Their proof uses an entropy-counting argument, which is arguably even simpler, but yields only constant probability of success.

The obvious open question is whether \cref{thm:escQLLL} can be generalized to the non-commuting case. The crucial (and only) place in our proof where commutativity is used and where the argument fails is \cref{lem:recursionsimple} (see also \cref{lem:recursion} in \cref{sec:details}). If the quantum algorithm is executed with non-commuting projectors, the present proof still shows that the algorithm terminates, i.e.\ the final measurement will project with high probability onto a subspace of terminated histories after the stated number of iterations. But we are not able to show that the state returned by the algorithm has low energy. Without commutativity, because measurements disturb quantum states, subsystems already checked at higher levels of the recursion may be messed up when fixing lower levels.

A further open question is whether Moser and Tardos' combinatorial proof \cite{MT10} of the Lovász Local Lemma for the more general, asymmetric case can be generalized to the quantum setting. It is interesting to note, that the dissipative algorithm of \cite{VWC09} is precisely the quantum analogue of Moser and Tardos' algorithm for the general, asymmetric Lovász Local Lemma written in the language of CP-maps. Thus, \cite{VWC09} already gives a way to prepare the ground state implied by the non-constructive QLLL \cite{AKS12}. What is missing is an argument supporting a polynomial-time convergence rate of the given CP-map. A first attempt in this direction for the case of commuting projectors has been made by the first and second author in \cite{CS11}. While the specific argument has an unresolved gap in the proof, the general framework based on dissipative CP-maps still appears as a promising approach and might lead to a complete proof in the future.

\section{Acknowledgements}
The authors wish to thank Julia Kempe, Or Sattath, and Robin Moser for valuable discussions.

MS is supported by Austrian SFB grant FoQuS F4014, TSC is supported by the Royal Society, and FV by EU grants
QUERG, and Austrian FWF SFB grants FoQuS and ViCoM.

\bibliography{quantum}{}

\newcommand{\etalchar}[1]{$^{#1}$}
\begin{thebibliography}{STV{\etalchar{+}}13}

\bibitem[AE11]{AE11}
Dorit Aharonov and Lior Eldar.
\newblock {On the complexity of Commuting Local Hamiltonians, and tight
  conditions for Topological Order in such systems}.
\newblock In {\em Foundations of Computer Science (FOCS), 2011 IEEE 52nd Annual
  Symposium on}, pages 334--343. IEEE, 2011.

\bibitem[AKS12]{AKS12}
A.~Ambainis, J.~Kempe, and O.~Sattath.
\newblock {A Quantum Lovasz Local Lemma}.
\newblock {\em Journal of the ACM (JACM)}, 59(5):24:1--24:24, November 2012.

\bibitem[AS00]{AS00}
Noga Alon and Joel~H. Spencer.
\newblock {\em The probabilistic method}.
\newblock Wiley-Interscience, 2000.

\bibitem[AS13]{AS13}
Itai Arad and Or~Sattath.
\newblock {A Constructive Quantum Lov\'asz Local Lemma for Commuting
  Projectors}.
\newblock {\em arXiv preprint arXiv:1310.7766}, 2013.

\bibitem[Bec91]{beck91}
J{\'o}zsef Beck.
\newblock {An Algorithmic Approach to the Lov{\'a}sz Local Lemma. I}.
\newblock {\em Random Structures \& Algorithms}, 2(4):343--365, 1991.

\bibitem[CS11]{CS11}
Toby~S. Cubitt and Martin Schwarz.
\newblock {A constructive commutative quantum Lov{\'a}sz Local Lemma, and
  beyond}.
\newblock {\em arXiv preprint arXiv:1112.1413}, 2011.

\bibitem[CT06]{CT06}
Thomas~M. Cover and Joy~A. Thomas.
\newblock {\em {Elements of Information Theory}}.
\newblock Wiley-Interscience, second edition, 2006.

\bibitem[EL75]{EL75}
Paul Erd\H{o}s and L{\'a}szl{\'o} Lov{\'a}sz.
\newblock Problems and results on 3-chromatic hypergraphs and some related
  questions.
\newblock {\em Infinite and finite sets}, II:609--627, 1975.

\bibitem[Mos09]{moser09}
Robin~A. Moser.
\newblock {A constructive proof of the Lov{\'a}sz Local Lemma}.
\newblock In {\em Proceedings of the 41st annual ACM Symposium on Theory Of
  Computing (STOC)}, pages 343--350. ACM, 2009.

\bibitem[Mos12]{moser12}
Robin~A. Moser.
\newblock {\em Exact Algorithms for Constraint Satisfaction Problems}.
\newblock PhD thesis, ETH Z\"{u}rich, 2012.

\bibitem[MT10]{MT10}
R.~A. Moser and G.~Tard{\'o}s.
\newblock {A constructive proof of the general Lov{\'a}sz Local Lemma}.
\newblock {\em Journal of the ACM (JACM)}, 57(2):1--15, 2010.

\bibitem[NC00]{NC00}
M.~Nielsen and I.~Chuang.
\newblock {\em {Quantum Computation and Quantum Information}}.
\newblock Cambridge University Press, Cambridge, 2000.

\bibitem[Osb12]{osborne12}
Tobias~J. Osborne.
\newblock Hamiltonian complexity.
\newblock {\em Reports on Progress in Physics}, 75(2):22001--22010, 2012.

\bibitem[Sch11]{schuch11}
Norbert Schuch.
\newblock Complexity of commuting hamiltonians on a square lattice of qubits,
  2011.

\bibitem[Spe77]{spencer1977}
Joel Spencer.
\newblock Asymptotic lower bounds for ramsey functions.
\newblock {\em Discrete Mathematics}, 20(0):69 -- 76, 1977.

\bibitem[STV{\etalchar{+}}13]{STVPGC13}
M.~Schwarz, K.~Temme, F.~Verstraete, D.~Perez-Garcia, and T.~S. Cubitt.
\newblock Preparing topological projected entangled pair states on a quantum
  computer.
\newblock {\em Physical Review A}, 88(032321), 2013.

\bibitem[VWC09]{VWC09}
F.~Verstraete, M.~M. Wolf, and J.~I. Cirac.
\newblock {Quantum computation, quantum state engineering, and quantum phase
  transitions driven by dissipation}.
\newblock {\em Nature Physics}, 5(9):633--636, 2009.

\bibitem[Wil13]{wilde13}
Mark~M. Wilde.
\newblock {\em {Quantum Information Theory}}.
\newblock Cambridge University Press, 2013.

\bibitem[Win99]{winter99}
Andreas Winter.
\newblock {\em {Coding Theorems of Quantum Information Theory}}.
\newblock PhD thesis, University of Bielefeld, 1999.

\end{thebibliography}
\bibliographystyle{alpha}

\appendix

\section{Detailed algorithm and proof} \label{sec:details}
We are now ready to introduce the more detailed quantum \cref{alg:quantummoser}. \Cref{alg:quantummoser}
is a manifestly unitary version of \cref{alg:moser} expanding all quantum registers necessary for bookkeeping,
unrolling the recursion into a unitary loop, and uncomputing auxiliary variables whenever
necessary for the rigorous argument. Furtheremore, we explicitly bound the number of iterations
necessary, such that with high probability all relevant histories in superposition have actually returned
from the (unrolled) recursion and terminated individually.

As already mentioned, our goal is to construct a unitary version of Moser's algorithm. Since projective measurements
are not unitary and can only be performed at the end of a standard quantum circuit, our approach is to
replace them by \emph{coherent measurements} \cite[Ch. 5.4]{wilde13}.
A coherent measurement of a binary observable $\{\Pi_i^0, \Pi_i^1\}$, with $\Pi_i^0+\Pi_i^1=\id$,
on a subsystem will correlate the state of a target qubit with the two possible measurement outcomes
in a unitary way.
This coherent measurement operation is performed by the following operator that is easily checked to be unitary:
\be \label{eq:coherent_measurement}
C_i = \Pi_i^0 \otimes \id + \Pi_i^1 \otimes X
\ee
where $X$ is the Pauli matrix $\sigma_x$.

\Cref{alg:quantummoser} operates on a quantum system consisting of register $W,R,L,F,term,
S,s,l,t,live$ summarized in \cref{fig:registers} at the end of the paper. We assume registers $W,R$ are initialized
in the completely mixed state. Register $W$ is the \emph{work register} in which our algorithm
will prepare a state $\sigma$ satisfying the symmetric QLLL conditions. Register $R$ is the source
of randomness that is fed into the work register by the algorithm appropriately. Register
$L$ is called the \emph{log register} holding an array of qubits $\ket{j_1,\dots,j_T}$ that store the binary coherent
measurement outcomes for a chosen projector $\Pi_i$ in each iteration of the algorithm.
Register $F$ is an array recording whether a recursion
level has terminated. While the information in this register is strictly redundant (relative to $L$), we find it
necessary to first compute and later uncompute the contents of this register to achieve an \emph{efficient}
unitary implementation of the algorithm that is provably correct up to the symmetric QLLL condition simultaneously.
Register $term$ is an array of qubits used to signal the termination of a measurement history in the
coherent superposition of histories. Once the qubit $term[l]$ is set to $\ket{1}$ in iteration $l$ in a particular history,
further iterations will just be idle in that history until the overall algorithm terminates. The stack
$S$ is an array of \emph{pairs} of registers, $proj$ and $nbr$. At recursion level $i$, register $S[i].proj$ refers to
the projector $\pi_{S[i].proj}$ being fixed in level $i$, where $S[i].nbr$ indicates the index (relative $S[i].nbr$)
of the \emph{neighboring} projector currently being verified ($0..k-1$). To simplify the presentation of
the algorithm, we treat the top-level of the recursion by pretending that some fiduciary clause had failed
that intersected with all clauses. In this way we can deal with the top-level
iteration just like with any other level. To this effect we initialize the content of register $S[0].proj = 0$
and define the special projector $\Pi_0$
to act non-trivially on \emph{all} $n$ qubits intersecting with \emph{all} projectors $\{\Pi_i\}_{1 \leq i \leq m}$.
This defines the top level of the recursion. Register $s$ is the stack pointer
referring to the current recursion level. Register $l$ is the log pointer, indicating the current iteration
of procedure \emph{iteration()} and the target qubit $L[l]$ for the coherent measurement in that iteration.
Register $t$ is the randomness pointer. It counts the number of failed measurements in a particular
measurement history and points to the next available block of $k$ random qubits starting at $R[tk]$.
Finally, register $live$ is a parameter to procedure \emph{iteration()} controlling whether operations
among the $W$ subsystem and the rest of the system should be performed ($live=1$) or skipped ($live=0$).
This is used to facilitate uncomputation of redundant information in the above registers.

We will now describe the operation of \cref{alg:quantummoser} in detail.
It consists of two procedures. The main procedure \emph{QLLL\_solver()} (\cref{proc:QLLLsolver}),
and procedure $\emph{iteration()}$ (\cref{proc:iteration}), which is called from \emph{QLLL\_solver()}.
\emph{QLLL\_solver()} starts by executing procedure \emph{iteration()} $T$ times in the forward and $T$ times
in the reverse direction, as indicated by the dagger symbol in \cref{line:reverse}. In the forward direction
procedure \emph{iteration()} (invoked with $live=1$)
applies a coherent
measurements of one of the $k$-QSAT projectors to the assignment register $W$ and stores the coherent
measurement outcome at the current position $l$ in the log register $L$. Based on the measurement outcome,
the stack and other bookkeeping registers are updated coherently as well. During the uncomputation phase
we invoke procedure \emph{iteration()} with parameter $live$ set to $0$ such that all bookkeeping registers are
uncomputed, \emph{except} the log $L$ itself as the coherent ``unmeasurements" are skipped.
Indeed, the contents of the large $F$ and $term$ registers has been completely uncomputed, as they
can be reconstructed from $L$ alone.
Note, that after the completion of the reverse iterations (before executing \cref{line:compress}),
all registers are back to their initial states, except the $W$, $L$, and $R$ registers.
Once all redundancy in the bookkeeping registers has been removed by uncomputation, procedure \emph{compress()}
compresses the $R,L,t$ registers as explained in more detail in the next section.
The function will return with the quantum state of register $t$ recomputed. Finally
a projective measurement on the subspace of histories with $t<N$ failed measurements
is performed, in which case the algorithm returns SUCCESS and the subsystem $W$ of quantum state
$(\id-P_{N})\sigma(\id-P_{N})$, or FAILURE otherwise.

We will now describe the procedure \emph{iteration()}. Unless the algorithm has terminated
(or the function is not called with $live=1$) each iteration of the algorithm performs exactly
one coherent measurement (\cref{line:cmeasurement}) and all necessary update actions on
the state variables to simulate the recursive procedure of Moser's algorithm.
Since the measurement is coherent, the \emph{execution splits into a superposition} of two possible measurement
outcomes whenever this line of the algorithm is executed, unless all projectors are classical.
In the case that the measurement \emph{fails} (and \emph{iteration()} is called with $live=1$) the
procedure \emph{swap\_and\_rotate()} (\cref{line:swapandrotate})
is invoked, denoted $R_i$ below.

We are free to restrict our analysis to one particular history $\ket{j_1,\dots,j_l}$
since the quantum state is just a superposition of all possible such histories. To see this,
we proof the following
\begin{lemma} \label{lem:superposition}
For any initial state
\begin{align}
\ket{\psi_0^{x,y}} &= \ket{x}^W \ket{y}^R\ket{0_1,\dots,0_T}^L \ket{0}^l{\ket{0}}^t{\ket{0}}^{F,S,s,live,term}
\end{align}
with randomly chosen bit strings $x,y$, the quantum state produced by \cref{alg:quantummoser}
after $T>0$ iterations has the following structure:
\begin{align}
\ket{\psi_T^{x,y}} &=\sum_{t=0}^T P_r^{\ox t} \sum_{\substack{j_1+\dots+j_T=t \\ j_i \in \{0,1\}}}  \ket{\varphi_{j_1,\dots,j_T}}^{W,R} \ket{j_1,\dots,j_T}^L \ket{T}^l{\ket{t}}\ket{z_{j_1,\dots,j_T}}^{F,S,s,live,term} \label{eq:finalform}
\\
&=\sum_{j_1,\dots,j_T \in \{0,1\}}  P_r^{\ox t_{j_1,\dots,j_T}} \ket{\varphi_{j_1,\dots,j_T}}^{W,R}\ket{j_1,\dots,j_T}^L \ket{T}^l{\ket{t_{j_1,\dots,j_T}}}\ket{z_{j_1,\dots,j_T}}
\end{align}
where $t_{j_1,\dots,j_T}=\sum_{i=1}^T j_i$, and
where $P_r^{\ox t}$ acts only non-trivially on the first $kt$ qubits of register $R$.
That is, the state can be written as a (non-uniform) superposition of
$2^T$ orthogonal states enumerating all $T$-bit computational basis states $\ket{j_1,\dots,j_T}$ in the $L$ register,
each of which is entangled with some quantum state $\ket{\varphi_{j_1,\dots,j_T}}$ in the $W$ and $R$ registers,
and \emph{computational basis states} in the $t,F,S,s,live$, and $term$ registers. Furthermore, the $R$-register
components of this state live in a subspace of rank at most $rk(P_r^{\otimes t})=r^t$, where $t=\sum_{i=1}^T j_i$.
\end{lemma}
\begin{proof}
The proof proceeds by induction over $T$. The initial state of the algorithm (i.e. $T=0$ iterations) is
\begin{align}
\ket{\psi_0} &= \ket{x}^W \ket{y}^R\ket{0_1,\dots,0_T}^L \ket{0}^l{\ket{0}}^t{\ket{0}}^{F,S,s,live,term}
\end{align}
We claim that after $1\leq l\leq T$ iterations the state has the following slightly more general structure:
\begin{align}
\ket{\psi_l} =\sum_{t=0}^l P_r^{\ox t} \sum_{\substack{j_1+\dots+j_l=t \\ j_i \in \{0,1\}}} \ket{\varphi_{j_1,\dots,j_T}}^{W,R} \ket{j_1,\dots,j_l,0_{l+1},\dots,0_T}^L \ket{l}{\ket{t}}\ket{z_{j_1,\dots,j_l}}^{F,S,s,live,term}  \label{eq:generalform}
\end{align}
Clearly, for $l=T$ the lemma follows. To prove the base case $l=1$, notice that after the first iteration the state evolves to
\begin{align}
\ket{\psi_1}=R_0 C_0 \ket{\psi_0} &= R_0 C_0 \ket{x,y}^{W,R}\ket{0_1,\dots,0_T}^L \ket{0}^l{\ket{0}}^t{\ket{0}}^{F,S,s,live,term} \\
&=R_0 \Pi_0^0 \ket{x,y}^{W,R}\ket{0_1,0_{2}\dots,0_T}^L \ket{1}^l{\ket{0}}^t\ket{z_{0}} \label{eq:appCM}\\
&+R_0 \Pi_0^1 \ket{x,y}^{W,R}\ket{1_1,0_{2}\dots,0_T}^L \ket{1}^l{\ket{1}}^t\ket{z_{1}} \nonumber \\
&=R_0 \ket{\varphi_{0}}^{W,R}\ket{0_1,0_{2}\dots,0_T}^L \ket{1}^l{\ket{0}}^t\ket{z_{0}} \label{eq:proj} \\
&+R_0 \ket{\varphi'_{1}}^{W,R}\ket{1_1,0_2\dots,0_T}^L \ket{1}^l{\ket{1}}^t\ket{z_{1}} \nonumber \\
&= \ket{\varphi_{0}}^{W,R}\ket{0_1,0_{2}\dots,0_T}^L \ket{1}^l{\ket{0}}^t\ket{z_{0}} \label{eq:appR} \\
&+P_r \ket{\varphi_{1}}^{W,R}\ket{1_1,0_2\dots,0_T}^L \ket{1}^l{\ket{1}}^t\ket{z_{1}} \nonumber \\
&= P_r^{\ox 0} \ket{\varphi_{0}}^{W,R}\ket{0_1,0_2,\dots,0_T}^L \ket{1}^l{\ket{0}}^t\ket{z_{0}} \label{eq:fiducial}\\
&+ P_r^{\ox 1} \ket{\varphi_{1}}^{W,R}\ket{1_1,0_2,\dots,0_T}^L \ket{1}^l{\ket{1}}^t\ket{z_{1}} \nonumber \\
&=\sum_{t=0}^1 P_r^{\ox t}\sum_{\substack{j_1=t \\ j_1 \in \{0,1\}}}  \ket{\varphi_{j_1}}^{W,R}\ket{j_1,0_{2}\dots,0_T}^L \ket{1}^l{\ket{t}}\ket{z_{j_1}} \label{eq:sum1}.
\end{align}
where in \cref{eq:appCM} we expand the coherent measurement $C_0$ using \cref{eq:coherent_measurement}. We denote classical
bookkeeping states in registers $F,S,s,L,live,term$ collectively as $\ket{z_i}$ henceforth. Note that the projectors
act on $W$ while $\id$ and $X$ act on qubit $l=0$ in $L$, respectively. We see that the state splits into a superposition of
two states, with orthogonal states in qubit $\ket{j_1}$. In \cref{eq:proj} we label the projected states by $\Pi_0^0 \ket{x,y}=\ket{\varphi_0}$, and $\Pi_0^1 \ket{x,y}=\ket{\varphi_1'}$. In \cref{eq:appR} we apply the \emph{swap\_and\_rotate()} operation $R_0$, which acts at the identity on the first term. On the second term, the projected qubits are swapped from the $W$ into the $R$ register and then rotated into the $P_r$ subspace, transforming the state into
\begin{equation}
\ket{\varphi_1} = U_0^{R_0}\cdot U_{\text{SWAP}}^{W_0R_0} \ket{\varphi_1'} = U_0^{R_0}\cdot U_{\text{SWAP}}^{W_0R_0}\Pi_0^1 \ket{x,y}=P_r^{W_0} U_0^{R_0}\cdot U_{\text{SWAP}}^{W_0R_0}\ket{x,y}.
\end{equation}
which follows from \cref{eq:livesinPr,eq:defRi}. Furthermore, this also implies that $\ket{\varphi_1}=P_r^{R_0} \ket{\varphi_1}$, so we are justified in explicitly extracting the projector $P_r$ in \cref{eq:appR}.
In \cref{eq:fiducial} we insert the fiducial projector $P_r^{\otimes 0} = \id$ in order to rewrite the equation into a sum of the desired structure in \cref{eq:sum1}. Thus, the state $\ket{\psi_1}$ has the required structure with $l=1$, which proves the base case.

In subsequent iterations, we denote the operations of \cref{alg:quantummoser} by operators $C_{j_1,\dots,j_l}$ (\emph{coherent measurement}), and $R_{j_1,\dots,j_l}$ (\emph{swap\_and\_rotate}), respectively. These are controlled by the content of the $L$ and $l$ registers. All further bookkeeping operations are to be considered to be part of $R_{j_1,\dots,j_l}$. We need to show that the state has the structure of \cref{eq:generalform} for all $l$. This is indeed the case, since
{\small
\begin{align}
\ket{\psi_{l+1}} &= R_{j_1,\dots,j_l} C_{j_1,\dots,j_l} \ket{\psi_l}\\
&= R_{j_1,\dots,j_l} C_{j_1,\dots,j_l} \sum_{t=0}^l P_r^{\ox t} \sum_{\substack{j_1+\dots+j_l=t \\ j_i \in \{0,1\}}} \ket{\varphi_{j_1,\dots,j_T}}^{W,R} \ket{j_1,\dots,j_l,0_{l+1},\dots,0_T}^L \ket{l}{\ket{t}}\ket{z_{j_1,\dots,j_l}}
\end{align}
\begin{align} \label{eq:appCM2}
=R_{j_1,\dots,j_l} \sum_{t=0}^l P_r^{\ox t} \sum_{\substack{j_1+\dots+j_l=t \\ j_i \in \{0,1\}}} (\Pi_{j_1,\dots,j_l}^0  &\ket{\varphi_{j_1,\dots,j_l}}^{W,R}\ket{j_1,\dots,j_l,0,0_{l+2}\dots,0_T}^L \ket{l+1}{\ket{t}}\ket{z_{j_1,\dots,j_l}} \nonumber\\
+\Pi_{j_1,\dots,j_l}^1 &\ket{\varphi_{j_1,\dots,j_l}}^{W,R} \ket{j_1,\dots,j_l,1,0_{l+1}\dots,0_T}^L \ket{l+1}{\ket{t+1}}\ket{z_{j_1,\dots,j_l}})
\end{align}
\begin{align} \label{eq:proj2}
=R_{j_1,\dots,j_l}\sum_{t=0}^l P_r^{\ox t} \sum_{\substack{j_1+\dots+j_l=t \\ j_i \in \{0,1\}}} (&\ket{\varphi_{j_1,\dots,j_l,0}}^{W,R}\ket{j_1,\dots,j_l,0,0_{l+2}\dots,0_T}^L \ket{l+1}{\ket{t}}\ket{z_{j_1,\dots,j_l}}\nonumber\\
+&\ket{\varphi'_{j_1,\dots,j_l,1}}^{W,R} \ket{j_1,\dots,j_l,1,0_{l+1}\dots,0_T}^L \ket{l+1}{\ket{t+1}}\ket{z_{j_1,\dots,j_l}})
\end{align}
\begin{align} \label{eq:appR2}
=\sum_{t=0}^l P_r^{\ox t} \sum_{\substack{j_1+\dots+j_l=t \\ j_i \in \{0,1\}}} (&\ket{\varphi_{j_1,\dots,j_l,0}}^{W,R}\ket{j_1,\dots,j_l,0,0_{l+2}\dots,0_T}^L \ket{l+1}{\ket{t}}\ket{z_{j_1,\dots,j_l}}\nonumber\\
+{P_r}^{R_{t}}&\ket{\varphi_{j_1,\dots,j_l,1}}^{W,R} \ket{j_1,\dots,j_l,1,0_{l+1}\dots,0_T}^L \ket{l+1}{\ket{t+1}}\ket{z_{j_1,\dots,j_l}})
\end{align}
\begin{align} \label{eq:final2}
=\sum_{t=0}^{l+1} P_r^{\ox t} \sum_{\substack{j_1+\dots+j_{l+1}=t \\ j_i \in \{0,1\}}} &\ket{\varphi_{j_1,\dots,j_l,j_{l+1}}}^{W,R}\ket{j_1,\dots,j_l,j_{l+1},0_{l+2}\dots,0_T}^L \ket{l+1}{\ket{t}}\ket{z_{j_1,\dots,j_l}}
\end{align}
}
\noindent where, again, in \cref{eq:appCM2} we expand the coherent measurement $C_{j_1,\dots,j_l}$ using \cref{eq:coherent_measurement}, where the projectors act on $W$ while $\id$ and $X$ act on qubit $l$ in $L$, respectively. We see that the state splits into a superposition of two states orthogonal in the state of this qubit. Register $l$ is increased by one in both states. In \cref{eq:proj2} we label the projected states by $\Pi_{j_1,\dots,j_l}^0 \ket{\varphi_{j_1,\dots,j_l}}=\ket{\varphi_{j_1,\dots,j_l,0}}$, and $\Pi_{j_1,\dots,j_l}^1 \ket{\varphi_{j_1,\dots,j_l}}=\ket{\varphi_{j_1,\dots,j_l,1}'}$. In \cref{eq:appR2} we apply the \emph{swap\_and\_rotate()} operation $R_{j_1,\dots,j_l}$, which acts at the identity on the first term.
On the second term, the projected qubits are swapped from the $W$ into the $R$ register and then rotated into the $P_r$ subspace, transforming the state into
\begin{equation}
\ket{\varphi_{j_1,\dots,j_l,1}} = U_i^{R_t}\cdot U_{\text{SWAP}}^{W_iR_t} \ket{\varphi_{j_1,\dots,j_l,1}'} = U_i^{R_t}\cdot U_{\text{SWAP}}^{W_iR_t}\Pi_0^1 \ket{\varphi_{j_1,\dots,j_l}}=P_r^{W_0} U_i^{R_t}\cdot U_{\text{SWAP}}^{W_iR_t}\ket{\varphi_{j_1,\dots,j_l}}.
\end{equation}
which follows from \cref{eq:livesinPr,eq:defRi}. Furthermore, this also implies that $\ket{\varphi_{j_1,\dots,j_l,1}}=P_r^{R_{t}} \ket{\varphi_{j_1,\dots,j_l,1}}$, thus we are justified in explicitly extracting the projector $P_r$ in \cref{eq:appR}.
Finally, in \cref{eq:final2} we rewrite the state by adding the binary index $j_{l+1}$ in the inner sum. Furthermore, we sum $t$ up to $l+1$ accommodating the additional measurement. Evidently, the state has now the form claimed for $\ket{\psi_{l+1}}$. By induction, the state has the required form of \cref{eq:finalform} for all $1 \leq l \leq T$, yielding the lemma.
\end{proof}

\subsection{Proof of \protect\cref{thm:escQLLL}}
\begin{proof}
By \cref{lem:superposition} we know that after $T$ iterations of \cref{alg:quantummoser} the state has the form
\begin{align}
\ket{\psi_T^{x,y}} &=\sum_{t=0}^T P_r^{\ox t} \sum_{\substack{j_1+\dots+j_T=t \\ j_i \in \{0,1\}}}  \ket{\varphi_{j_1,\dots,j_T}}^{W,R} \ket{j_1,\dots,j_T}^L \ket{T}^l{\ket{t}}\ket{z_{j_1,\dots,j_T}}^{F,S,s,live,term} \label{eq:finalform2}
\end{align}
After uncomputing the redundant registers, this simplifies to
\begin{align}
\ket{\psi_U^{x,y}} &=\sum_{t=0}^T P_r^{\ox t} \sum_{\substack{j_1+\dots+j_T=t \\ j_i \in \{0,1\}}}  \ket{\varphi_{j_1,\dots,j_T}}^{W,R} \ket{j_1,\dots,j_T}^L \ket{T}^l{\ket{t}}\ket{0}^{F,S,l,s,live,term} \label{eq:finalform3}
\end{align}
One way to view state $\ket{\psi_U^{x,y}}$ is as a superposition of all possible measurement histories $j_1,\dots,j_T$, which were the result if we had performed projective rather than coherent measurements. By the \emph{principle of deferred measurement} \cite{NC00}, we can still measure all qubits in $L$ to project onto one of these histories.  Consequently, we call each term in the sum of \cref{eq:finalform2} a \emph{history} and identify histories by the outcomes $\ket{j_1,\dots,j_T}$ in register $L$.

Let us make a few observations about each history $\ket{j_1,\dots,j_T}^L$.
If $j_1,\dots,j_T$ contains $t$ failed measurement outcomes, we know from \cref{lem:superposition}
that \cref{alg:quantummoser} has
projected the first $t$ blocks of $k$ qubits in $R$ into the subspace $P_r^{\ox t}$.
\Cref{line:forcedtermination} enforces that $t \leq N \leq (T-m)/d$,
i.e. a maximum number $N$ of failed measurements, which we will choose later on. Thus by terminating execution once the
maximum admissible number of $N$ failed measurements has been reached, we
accept that some histories in superposition in $\ket{\psi_T^{x,y}}$ may not have returned from the recursion.
On the other hand, for all histories with $t < N$ it is clear that they must have returned
to the top-level of the recursion and terminated at iteration $T=m+dt$, since to the $m$ top-level measurements
exactly $d$ more measurements are added for each of the $t$ failed outcomes. Therefore, within the $T$
bits of $L$ at most $t$ bits are in state $\ket{1}$. Thus the Shannon entropy of
bit string $j_1,\dots,j_T$ relative to $t$ is at most
$\log{m+dt \choose t} \leq m+\log{dt \choose t} \leq m + \log{(\frac{det}{t})^{t}} = m + t\log(de)$
bits. By encoding $L$ by the \emph{index} of $j_1,\dots,j_T$ in the lexicographically ordered set
of bit strings of length $T$ with $t$ ones we can achieve compression of $L$ to the above bound,
relative to $t$ \cite[Ch. 13.2]{CT06}. This classical compression is performed reversibly in \cref{line:compress} by procedure $compress(j_1,\dots,j_T,t)$
for each history $\ket{j_1,\dots,j_T}$, in superposition.\footnote{Note a minor technicality: at the instant \emph{compress(L,t)} is invoked, the $t$ register has actually been uncomputed (like all other auxiliary variables) and must be recomputed within the function by simply counting the $t \leq N$ ones in each $\ket{j_1,\dots,j_T}$. It could also have been copied before uncomputation.
The recomputed value of $t$ remains in the register as the function returns as the compression $\ket{j_1,\dots,j_T}$ is relative to $\ket{t}$.}
Let us denote
the state after the compression as
\begin{align}
\ket{\psi_C^{x,y}} &= U_{\text{compress}} \ket{\psi_U^{x,y}} = \nonumber \\
&=\sum_{t=0}^T P_r^{\ox t} \sum_{\substack{j_1+\dots+j_T=t \\ j_i \in \{0,1\}}}  \ket{\varphi_{j_1,\dots,j_T}}^{W,R} \left(\ket{L_{j_1,\dots,j_T}}\ket{0}^{\otimes(T-m-t\log(de))}\right)^L {\ket{t}}\ket{0}
\end{align}
where $L_{j_1,\dots,j_T}=compress(j_1,\dots,j_T,t)$.
We formalize our knowledge about $\ket{\psi_C^{x,y}}$ by constructing a projector $P_{M}$ onto the subspace
with $t \geq M$:
\be
P_{M} = \underbrace{\id^{\otimes n}}_W
      \otimes \underbrace{ P_r^{\ox M} \otimes \id^{\otimes (N-M)k} }_R
			\otimes \underbrace{\id^{\otimes m+M\log(de)} \otimes (\ket{0}\bra{0})^{T-m-{M}\log(de)}}_L
			\otimes \underbrace{\left(\sum_{\tau=M}^N \ket{\tau}\bra{\tau}\right)}_t
			\otimes \underbrace{\ket{0}\bra{0}}_{F,S,s,l,live} 
\ee
We now show that for $M > \Omega\left(\frac{m+\log(N)}{k-\log(der)}\right)$ the probability of successfully projecting
the state
\be
\rho_C = \frac{1}{2^{n+Nk}}\sum_{x=0}^{2^n-1}\sum_{y=0}^{2^{Nk}-1} \ket{\psi_C^{x,y}}\bra{\psi_C^{x,y}}
\ee
i.e. $\ket{\psi_C^{x,y}}$ mixed over all $x,y$, onto $P_{M}$ is very low.
Clearly, mixing over $x,y$ injects $n+Nk$ bits of initial entropy.
Let $V=U_{\text{compress}}U_0^{\dag T} U_1^T$, and since $\ket{\psi_C^{x,y}}= V \ket{\psi_0^{x,y}}$, we have
\begin{align}
\rho_C &= \frac{1}{2^{n+Nk}}\sum_{x=0}^{2^n-1}\sum_{y=0}^{2^{Nk}-1} V\ket{\psi_0^{x,y}}\bra{\psi_0^{x,y}}V^\dag \\
&= \frac{1}{2^{n+Nk}}\sum_{x=0}^{2^n-1}\sum_{y=0}^{2^{Nk}-1} V(\ket{x}\bra{x}^W \ket{y}\bra{y}^R{\ket{0}\bra{0}}^{L,l,t,F,S,s,live,term})V^\dag \\
& = \frac{1}{2^{n+Nk}} V(\id\otimes\ket{0}\bra{0}) V^{\dag}
\end{align}
We now apply the following simple special case of the \emph{strong converse of the typical subspace theorem} \cite{winter99}
to get an upper bound for the overlap of $\rho_C$ with $P_{M}$. Note, that the following bound for this special case is
slightly stronger than the original bound of \cite{winter99}.
\begin{lemma} \label{lem:subentropic2}
Let $Q$ be a projector on any subspace of $(\mathbbm{C}^2)^{\otimes (n+m)}$ of dimension at most $2^{nR}$,
where $R < 1$ is fixed and $\frac{\id}{2^n} \otimes (\ket{0}\bra{0})^m$ a completely mixed state with pure ancillas. Then,
\be \label{strongerbound}
  \tr\left(Q \;\left(\frac{\id}{2^n} \otimes (\ket{0}\bra{0})^{\otimes m}\right)\right) \leq \tr\left(Q \frac{\id}{2^n}\right)= 2^{-n} \tr(Q) \leq 2^{-n+nR}
\ee
\end{lemma}
\begin{proof}
The proof is immediate in \cref{strongerbound}.
\end{proof}
\noindent Thus we achieve the bound
\begin{align} \label{eq:exponentiallysmall}
\tr(P_{M} \rho_C) &= 2^{-(n+Nk)} \tr(P_{M} V (\id\otimes{\ket{0}\bra{0}}) V^{\dag } )  \\
& \leq 2^{-(n+Nk)} \tr(P_{M} ) \\
& \leq 2^{m+\log(N)-M(k-\log(r)-\log(de))}\\
& \leq 2^{m+\log(N)-M(k-\log(der))}
\end{align}
On the other hand when $N=M$ we conclude, that the projector $(\id-P_{N})$ onto histories with
$t < N$ has overlap exponentially close to $1$ with $\rho_S$. In other words, \cref{alg:quantummoser}
returns SUCCESS in \cref{line:success} with
\be
Pr[SUCCESS, \sigma] \geq 1 - 2^{m+\log(N)-N(k-\log(der))}
\ee
It follows that choosing $N$ such that
\be \label{eq:implicitttilde}
N \geq \frac{m+\log(\frac{1}{\varepsilon})}{k-\log(der)}+\frac{\log(N)}{k-\log(der)}
\ee
suffices to push the error below $1-\varepsilon$.
But this bound for $N$ is not yet explicit. To get an explicit bound we
define
$c=(k-\log(der))^{-1}$, and $d=\frac{m+\log(\frac{1}{\varepsilon})}{k-\log(der)}$, and set (\cref{line:ttilde})
\be
N = d + 3c(\log(d)+1)
\ee
or, equivalently but more verbosely,
\be
N =  \frac{m+\log(\frac{1}{\varepsilon})}{k-\log(der)} + \frac{3(\log(\frac{m+\log(\frac{1}{\varepsilon})}{k-\log(der)})+1)}{k-\log(der)}
\ee
satisfying \cref{eq:implicitttilde} as shown in \cref{lem:explicitttilde} in the appendix.
Thus we conclude that after $T = m+Nd$ (\cref{line:T}) iterations of \cref{alg:quantummoser},
\be
Pr[SUCCESS, \sigma] \geq 1 - \varepsilon
\ee
as claimed.

In summary, we have shown that either the algorithm achieves a
compression of its state below the entropy of the initial state, which is unlikely, or in all histories
in superposition the number of failed measurements is upper bounded by $N$ and thus the histories must have terminated in the
state returned by \cref{alg:quantummoser}. Furthermore, the probability of the latter outcome can be pushed exponentially close to $1$. All that is left to show is that the state, once projected into the $(\1-P_{N})$ subspace, satisfies the
symmetric QLLL condition. By \cref{lem:recursion} shown below we know that each terminated history $j_1,\dots,j_T$ is correlated to a state $\ket{\varphi_{j_1,\dots,j_T}}$ with energy exactly zero. Thus it follows that the $W$ subsystem of state $(\id-P_{N})\rho_C(\id-P_{N})$ returned by \cref{alg:quantummoser} is just a mixture of zero energy states and has thus energy zero itself, which completes the proof, i.e. formally let
\begin{align}
\rho_P = \frac{(\1-P_{N})\rho_C (\1-P_{N})}{1-\tr(P_{N}\rho_C)}
\end{align}
where the denominator is exponentially close to $1$ due to \cref{eq:exponentiallysmall}. Then, expanding the definition of $\rho_C$ and recognizing that the projector on $\1-P_{N}$ just changes the upper bound of the sum over $t$ (and $t'$) from $T$ to $N$, we have
\begin{align}
\tr_{\overline{W,R}}(\rho_P)\propto& \tr_{\overline{W,R}} ((\1-P_{N})\rho_C(\1-P_{N})) \\
=&\tr_{\overline{W,R}}((\1-P_{N})\ket{\psi_S^{x,y}}\bra{\psi_S^{x,y}}(\1-P_{N}))\\
=&\frac{1}{2^{n+Nk}}\sum_{x=0}^{2^n-1}\sum_{y=0}^{2^{Nk}-1} \sum_{t=0}^{N}\sum_{\substack{j_1+\dots+j_{N}=t \\ j_i \in \{0,1\}}} \sum_{t'=0}^{N}\sum_{\substack{j'_1+\dots+j'_{N}=t' \\ j'_i \in \{0,1\}}}
\\ &\tr_{\overline{W,R}}(\ket{\varphi_{j_1,\dots,j_{N}}}\bra{\varphi_{j'_1,\dots,j'_{N}}}^{W,R}
(\ket{L'_{j_1,\dots,j_{N}}}\bra{L'_{j'_1,\dots,j'_{N}}}\ket{0}\bra{0})^L {\ket{t}}{\bra{t'}}\ket{0}\bra{0}) \nonumber \\
=&\frac{1}{2^{n+Nk}}\sum_{x=0}^{2^n-1}\sum_{y=0}^{2^{Nk}-1} \sum_{t=0}^{N}\sum_{\substack{j_1+\dots+j_{N}=t \\ j_i \in \{0,1\}}} \ket{\varphi_{j_1,\dots,j_{N}}}\bra{\varphi_{j_1,\dots,j_{N}}}^{W,R}  \label{eq:tensordist} \\
&\tr(\ket{L'_{j_1,\dots,j_{N}}}\bra{L'_{j_1,\dots,j_{N}}})\tr(\ket{0}\bra{0}) \tr({\ket{t}}{\bra{t}})\tr(\ket{0}\bra{0}) \nonumber \\
=&\frac{1}{2^{n+Nk}}\sum_{x=0}^{2^n-1}\sum_{y=0}^{2^{Nk}-1} \sum_{t=0}^{N}\sum_{\substack{j_1+\dots+j_{N}=t \\ j_i \in \{0,1\}}} \ket{\varphi_{j_1,\dots,j_{N}}}\bra{\varphi_{j_1,\dots,j_{N}}}^{W,R}  \label{eq:mixture}
\end{align}
\noindent where in \cref{eq:tensordist} we distribute the partial trace over the tensor factors. Since
orthogonal states evaluate to zero in each factor, only terms of factors with matching indices survive
in the sum, in which case these factors happen to be projectors of trace $1$. Thus \cref{eq:mixture} follows, which is clearly
a mixture of states $\ket{\varphi_{j_1,\dots,j_{N}}}$ as claimed. Since every $\ket{\varphi_{j_1,\dots,j_{N}}}$ is a state associated
to a terminated history,
we know the recursion of \cref{alg:quantummoser} has returned to the top level, in which all $m$ initial projectors
$\Pi_i$ have been measured. Thus by \cref{lem:recursion} we conclude that $\Pi_i^1\ket{\varphi_{j_1,\dots,j_{N}}}=0$
for all histories $j_1,\dots,j_{N}$.
\end{proof}
\noindent Note that the following lemma is the crucial (and only) place in the proof where commutativity of the projectors $\{\Pi_i\}$ is assumed.
\begin{lemma} \label{lem:recursion}
According to \cref{lem:superposition}, consider a history $\ket{j_1,\dots,j_l}$ in the superposition after $l$ coherent measurements
\be
\ket{\psi_l} = \ket{\varphi_{j_1,\dots,j_l}}^{W,R} \ket{j_1,\dots,j_l,0_{l+1},\dots,0_T}^L \ket{l}{\ket{t_{j_1,\dots,j_l}}}\ket{z_{j_1,\dots,j_l}}^{F,S,s,live,term}
\ee
where the last measurement has failed, i.e $j_l=1$. In this state \cref{alg:quantummoser} has started a new recursion level
and will coherently measure all projectors $\Pi_k \subseteq \Gamma^+(\Pi^1_{j_1,\dots,j_l})$ in subsequent iterations. For some iteration $m \geq l+k$, let
\be
\ket{\psi_m} = \ket{\varphi_{j_1,\dots,j_m}}^{W,R} \ket{j_1,\dots,j_m,0_{m+1},\dots,0_T}^L \ket{m}{\ket{t_{j_1,\dots,j_m}}}\ket{z_{j_1,\dots,j_m}}^{F,S,s,live,term}
\ee
be an extension of history $\ket{\psi_l}$ (i.e. with matching $j_1,\dots,j_l$) where \cref{alg:quantummoser} has just returned from that recursion. 
Then
\begin{enumerate}
\item all satisfied projectors $\Pi_i^1$ stay satisfied, i.e. if $\Pi_i^1 \ket{\varphi_{j_1,\dots,j_l}} = 0$,
      then also $\Pi_i^1 \ket{\varphi_{j_1,\dots,j_m}} = 0$. \label{it:other_projectors_satisfied2}
\item the originally unsatisfied projector is now satisfied, i.e. $\Pi_{j_1,\dots,j_l}^1\ket{\varphi_{j_1,\dots,j_m}}=0$.
      \label{it:current_projector_satisfied2}
\end{enumerate}
\end{lemma}
\begin{proof}
We first prove \cref{it:other_projectors_satisfied2} by induction on the
stack level $s$ of \cref{alg:quantummoser}, starting from the deepest level,
which must exist because the algorithm returns by assumption.\footnote{
In the main algorithm we apply this lemma only to histories in the subspace $(\id-P_{N})$,
where we have already shown that all histories terminate.} The recursive call can only return
if all $\Pi^1_i \in \Gamma^+(\Pi(s))$ are satisfied, i.e. $\Pi^1_i \ket{\varphi_{j_1,\dots,j_m}}=0$.
For all $\Pi^1_q \nsubseteq \Gamma^+(\Pi(s))$ with $\Pi^1_q \ket{\varphi_{j_1,\dots,j_l}}=0$,
we have
\begin{align}
\Pi_q^1 \ket{\varphi_{j_1,\dots,j_m}}\ket{\xi'} &= \Pi_q^1 \prod_{i \in \Gamma^+} \Pi_i^0 R_{j_1,\dots,j_l} \Pi_{j_1,\dots,j_l}^1 \ket{\varphi_{j_1,\dots,j_l}}\ket{\xi} \\
&= \prod_{i \in \Gamma^+} \Pi_i^0 R_{j_1,\dots,j_l} \Pi_{j_1,\dots,j_l}^1 \Pi_q^1 \ket{\varphi_{j_1,\dots,j_l}}\ket{\xi}  = 0
\end{align}
where we expand $\ket{\varphi_{j_1,\dots,j_m}}$ by the action of \cref{alg:quantummoser} in the first equality, where $\ket{\xi'}$ represents the state of subsystems other than $W,R$.
In the second equality we commute $\Pi^1_q$ through, which is possible, because $\Pi^1_q$ and all $\Pi_i$ commute
by assumption, and $\Pi^1_q$ and $R_{j_1,\dots,j_l} \Pi_{j_1,\dots,j_l}^1$ commute
because they act on different subsystems. Finally, the last equation follows since $\Pi_q^1 \ket{\varphi_{j_1,\dots,j_l}}=0$ is the precondition under which we need to prove \cref{it:other_projectors_satisfied2} of \cref{lem:recursion}. This proves the base case of the induction. The inductive step follows
from exactly the same arguments, thus \cref{it:other_projectors_satisfied2} follows.
To show \cref{it:current_projector_satisfied2} of \cref{lem:recursion}, it suffices to note that $\Pi_{j_1,\dots,j_l}^1 \in \Gamma^+(\Pi_{j_1,\dots,j_l}^1)$,
thus $\Pi_{j_1,\dots,j_l}^1\ket{\varphi_{j_1,\dots,j_m}}=0$ is true since the algorithm just returned from a recursive call on a failed measurement of $\Pi_{j_1,\dots,j_l}^1$ by assumption: i.e. in the iterations $<m$ just before the algorithm has returned, all
$\Pi_q^1 \in \Gamma^+(\Pi_{j_1,\dots,j_l}^1)$ had been measured to be satisfied (or fixed and then satisfed by \cref{it:other_projectors_satisfied2}).
Since all $\Pi_q^1$ commute, this implies $\Pi_{j_1,\dots,j_l}^1\ket{\varphi_{j_1,\dots,j_m}}=0$.
\end{proof}

\subsection{Upper bound on $N$}
In this section we compute an upper bound for $N$ defined implicitly by
\be \label{eq:ttilde}
N = \frac{\log(N)}{k-\log(der)} + \frac{m+\log(\frac{1}{\varepsilon})}{k-\log(der)}
\ee
\begin{lemma} \label{lem:explicitttilde}

Define $a=(k-\log(der))^{-1}$, $b=\frac{m+\log(\frac{1}{\varepsilon})}{k-\log(der)}$, and $N=t+a\log(t)$, then
\be
N
\leq
b + a (\log(a+1) + \log(b + a\log(a+1) ))
\leq
b + 3a(\log(b)+1)
\ee
\end{lemma}
\begin{proof}
We start with \cref{eq:ttilde} as the implicit definition of $N$ to
derive the upper bound. Expanding the substitutions reduces \cref{eq:ttilde} to
\be
t+a\log(t) = a \log(t+a\log(t)) + b
\ee
Then we bound $\log(t)\leq t$ coarsely on the r.h.s., which yields
\be
t+a\log(t) \leq a \log(t(a+1)) + b
\ee
\be
t+a\log(t) \leq a \log(a+1)+a\log(t) + b
\ee
\be
t \leq a \log(a+1) + b
\ee
Thus
\be
N\leq b + a (\log(a+1) + \log(b + a\log(a+1) ))
\ee
which can be evaluated explicitly. Relaxing the bound further yields
\be
N \leq b + 3a(\log(b)+1)
\ee
\end{proof}
\noindent As asymptotic bounds we also have
$N \leq \frac{m+\log(\frac{1}{\varepsilon})}{k-\log(de)} + O(\log(m+\log(\frac{1}{\varepsilon})))$
or
$N \leq O(m+\log(\frac{1}{\varepsilon}))$.

\begin{algorithm}[!hbtp]
{\small
  \caption{Quantum information-theoretic QLLL solver} \label{alg:quantummoser}
	\begin{algorithmic}[1]
		\Procedure{\textnormal{QLLL\_solver}\label{proc:QLLLsolver}}{}
			\State $a:=1/\log(k-de)$
			\State $b:=(m+\log(1/\varepsilon))/\log(k-de)$
			\State $N :=  b + 3a(\log(b)+1)$ \label{line:ttilde}
			\State $T := m+Nd$ \label{line:T}
			\For{$l := 0$ to $T-1$} \label{line:block1:start}
				\State iteration($live=1$)
			\EndFor
			\For{$l := T-1$ to $0$}
				\State iteration$^\dag$($live=0$) \label{line:reverse}
			\EndFor \label{line:block1:end}
			\State compress($L,t$) \label{line:compress}
			\If{measure($\{P_{N},\id-P_{N}\}$)=$(\id-P_{N})$}
			  \State return SUCCESS, $W$ \label{line:success}
			\Else
			  \State return FAILURE
			\EndIf
		\EndProcedure

		\Procedure{\textnormal{iteration} \label{proc:iteration} }{live}
		  \If{not $term[l]$}
				\If{$live$}
					 \State $L[l] \gets \text{measure\_coherently}(\Gamma^+(S[s].proj,S[s].nbr))$ \label{line:cmeasurement}
				\EndIf
				\If{$L[l]$} \label{line:failed}
					\If{$live$}
						\State swap\_and\_rotate$(\Gamma^+(S[s].proj,S[s].nbr),R[tk])$ \label{line:swapandrotate}
					\EndIf
					\State $t \gets t+1$
					\If{$t=N$}
						\State $term[l+1] \gets term[l+1]+1$  \label{line:forcedtermination}
					\EndIf
					\State $S[s+1].proj \gets S[s+1].proj + \Gamma^+(S[s].proj,S[s].nbr)$
					\State $s \gets s+1$ \label{line:enterrecursion}
				\Else
					\If{$s=0$}
						\State $S[s].nbr \gets S[s].nbr+1 \mod m$
						\If{$[s].nbr=0$}
							\State $term[l+1] \gets term[l+1]+1$
						\EndIf
					\Else
						\State $S[s].nbr \gets S[s].nbr+1 \mod k$
						\If{$S[s].nbr=0$}
							\State $F[l] \gets F[l]+1$
						\EndIf
					\EndIf
					\If{$F[l]$}
						\State $s \gets s-1$ \label{line:returnrecursion}
						\State $S[s+1].proj \gets S[s+1].proj - \Gamma^+(S[s].proj,S[s].nbr)$
					\EndIf
				\EndIf
			\Else
				\State $term[l+1] \gets term[l+1]+1$
			\EndIf
		\EndProcedure
	\end{algorithmic}
}
\end{algorithm}

\begin{table}[!htbp] \label{tab:registers}
\begin{tabular}{c|l|c|l|p{4.5cm}}
subsystem & description & size (qubits) & initial value & comment \\
\hline
$W$    & work register       & $n$       & $\id/2^n$     & random initial assignment \\
$R$    & randomness register & $Nk$      & $\id/2^{Nk}$  & source of entropy \\
$L$    & recursion log register&$m+Nd$      & $\ket{0\dots,0}$ & indicates a failed measurements and thus the start of recursion\\
$F$    & return flag register& $m+Nd$       & $\ket{0\dots,0}$ & indicates return from recursion\\
$term$ & termination register& $m+Nd$       & $\ket{0\dots,0}$ & indicates \emph{no further operations} need to be performed\\
$S$    & stack register 		 & $2\log(N)\log(m)$ & $\ket{0,0}\dots\ket{0,0}$ & $\log(N)$ pairs of registers labeled $(S[i].proj,S[i].nbr)$ used to indicate the projector we're fixing and the current neighbor we're checking\\
$s$    & stack pointer       & $\log(N)$ & $\ket{0}$     & indicates the recursion level. \\
$l$    & log pointer         & $\log(N)$ & $\ket{0}$     & indicates the next empty record. \\
$t$    & randomness pointer  & $\log(N)$ & $\ket{0}$     & indicates the next available block of $k$ random bits, also the number of failed coherent measurements. \\
$live$ & modify $W,R$?  & $1$       & $\ket{0}$     & indicates if changes to $W,R$ are executed ($1$) or skipped ($0$).
\end{tabular}
\caption{The quantum registers and the initial state of \protect\cref{alg:quantummoser}} \label{fig:registers}
\end{table}

\end{document}